\documentclass[11pt]{article}
\usepackage{graphicx}    
\usepackage{amsmath}
\usepackage{amsfonts}
\usepackage{hyperref}
\usepackage{amssymb}
\usepackage{amsthm}
\usepackage{latexsym}
\usepackage{enumerate}
\usepackage{subcaption}
\usepackage[margin=1in]{geometry}
\usepackage{bold-extra}

\newtheorem{theorem}{Theorem}[section]

\newtheorem{problem}{Problem}[section]

\theoremstyle{definition}

\graphicspath{{figs/}}

{\makeatletter
 \gdef\xxxmark{%
   \expandafter\ifx\csname @mpargs\endcsname\relax 
     \expandafter\ifx\csname @captype\endcsname\relax 
       \marginpar{xxx}
     \else
       xxx 
     \fi
   \else
     xxx 
   \fi}
 \gdef\xxx{\@ifnextchar[\xxx@lab\xxx@nolab}
 \long\gdef\xxx@lab[#1]#2{\textbf{[\xxxmark #2 ---{\sc #1}]}}
 \long\gdef\xxx@nolab#1{\textbf{[\xxxmark #1]}}
}

\def\andlinebreak{\end{tabular}\linebreak\begin{tabular}[t]{c}}

\newcommand{\NP}{\textrm{NP}}
\newcommand{\PSPACE}{\textrm{PSPACE}}

\newcommand{\PROBLEM}[1]{#1}

\let\epsilon=\varepsilon

\title{Losing at Checkers is Hard}

\author{
Jeffrey Bosboom\thanks{MIT Computer Science and Artificial Intelligence
  Laboratory, 32 Vassar Street, Cambridge, MA 02139, USA,
  \protect\url{{jbosboom,edemaine,mdemaine,jaysonl}@mit.edu}}
\and
Spencer Congero\thanks{Department of Electrical and Computer Engineering,
  University of California, San Diego, \protect\url{scongero@ucsd.edu}}
\and
Erik D.\ Demaine\footnotemark[1]
\andlinebreak
Martin L. Demaine\footnotemark[1]
\and
Jayson Lynch\footnotemark[1]
}

\date{}

\begin{document}
\maketitle

\begin{abstract}
We prove computational intractability of variants of checkers:
(1) deciding whether there is a move that forces the other player to win in one move is NP-complete;
(2) checkers where players must always be able to jump on their turn is PSPACE-complete; and
(3) cooperative versions of (1) and (2) are NP-complete.
We also give cooperative checkers puzzles whose solutions are the letters of the alphabet.
\end{abstract}

\section{Introduction}


Winning isn't for everyone.  Sometimes, you want to let your little sibling /
nibling / child win, so that they can feel the satisfaction, encouraging
them to play more games.  How do you play to lose?  This idea is formalized by
the notion of the \emph{mis\`ere} form of a game, where the rules are identical
except that the goal is to lose (have the other player win) according to the
normal rules \cite{Berlekamp-Conway-Guy-2001,Conway-1976-ONAG}.
Although your little relative may \emph{want} to win, it's probably
unreasonable to assume that they play optimally, so in the mis\`ere game,
we not only aim to lose but also to prevent the opponent from losing.

Checkers (in American English, ``draughts'' in British English) is a classic
board game that many people today grow up with.  Although there are several
variations, the main ruleset played today is known as American checkers and
English draughts, and has had a World Championship since the 1840s
\cite{Wikipedia}.  Computationally, the game is \emph{weakly solved}, meaning
that its outcome from the initial board configuration (only) has been computed
\cite{Schaeffer1518}: if both players play optimally to win, the outcome will
be a draw.

But what if both players play optimally to lose?  (Again, your little relative
may not \emph{aim} to lose, but you want to lose no matter what they do, so
this ``worst-case'' mathematical model more accurately represents the goal.)
The mis\`ere form of checkers is known as \emph{suicide checkers},
\emph{anti-checkers}, \emph{giveaway checkers}, and \emph{losing draughts}
\cite{Wikipedia}.  This game is less well studied.  No one knows whether the
outcome is again a draw or one of the players successfully losing.  There are
a few computer players such as Suicidal Cake \cite{SuicidalCake} that
occasionally compete against each other.

\paragraph{Complexity of winning.}
In this paper, we consider a few checkers variants, in particular suicidal
checkers, from a \emph{computational complexity} perspective.
Computational complexity allows us to analyze the best computer algorithm to
optimally play a game, and in particular, how the time and memory required by
that algorithm must \emph{grow} with the game size.  For this analysis to make
sense, we need to generalize the game beyond any ``constant'' size such as
12 pieces on an $8 \times 8$ board.  The standard generalization for checkers
is to consider an $n \times n$ board, and an arbitrary position of pieces
on that board (imagining that we reached this board configuration by several
moves from a hypothetical initial position).  The computational problem is to
decide who wins from this configuration assuming optimal play.

In this natural generalization, Robson \cite{robson1984n} proved that normal
checkers is \emph{EXPTIME-complete}, meaning in particular that any correct
algorithm must use time that grows exponentially fast as $n$ increases ---
roughly needing $c^n$ time for some constant $c > 1$.  EXPTIME-completeness is
a precise notion we will not formalize here, but it also implies that checkers
is one of the hardest problems that require exponential time, in the sense
that all problems solvable in exponential time (which is most problems we
typically care about) can be converted into an equivalent game of checkers.
So maybe all that time spent playing checkers as a kid wasn't a waste\dots\
(There are many such complexity results about many different games; see, e.g.,
\cite{GPC,6.890}.)

One of the checkers variants we study in this paper adds one rule: every move
must capture at least one piece.  Thus, a player loses if they cannot make a
jumping move.  This \emph{always-jumping} checkers is not very interesting
from the usual starting configuration (where no jumps are possible),
but it is interesting from an arbitrary board configuration.
In Section~\ref{always-jumping}, we show that always-jumping checkers is
\emph{PSPACE-complete} --- a complexity class somewhere in between ``easy''
problems (P) and EXPTIME-complete problems.  Formally, always-jumping
can be solved using polynomial space (memory) --- something we don't know how
to do for normal checkers, where many moves can be made between captures ---
and always-jumping checkers is among the hardest such problems.
This result suggests that always-jumping checkers is somewhat easier than
normal checkers, but not by much.

\paragraph{Complexity of losing.}
But what is the computational complexity of suicide checkers?
This question remains unsolved.

In this paper, we study a closely related problem: what is the computational
complexity of deciding whether you can lose the game ``in one move'', that is,
make a move that forces your opponent to win in a single move?
We call this the \emph{lose-in-1} variation of a game.
In Section~\ref{lose1}, we prove that lose-in-1 checkers is \emph{NP-complete}
--- a complexity class somewhere in between ``easy''
problems (P) and PSPACE-complete problems.  In particular, assuming
a famous conjecture (the Exponential Time Hypothesis),
NP-complete problems (and thus also PSPACE-complete problems)
require nearly exponential time ---
roughly $2^{n^\epsilon}$ for some $\epsilon > 0$.

\label{mate-in-1}
Hardness of losing checkers in one move is particularly surprising because
\emph{winning} checkers in one move (\emph{mate-in-1}) is known to be easy.
An early paper about the complexity of normal checkers
\cite{fraenkel1978complexity} gave an efficient algorithm for deciding
mate-in-1 checkers, running in time linear in the number of remaining pieces.
Namely, for each of your pieces, draw a graph whose vertices represent
reachable positions by jumping, and whose edges represent (eventually)
jumpable pieces.
Then, assuming all opponent pieces appear as edges in this graph,
winning the game is equivalent to finding an \emph{Eulerian path} in the graph
(a path that visits each edge exactly once)
that starts at the piece's initial position.
Finding Eulerian paths goes back to Euler in 1736, though the problem was not
fully solved until 1873 by Hierholzer \cite{Hierholzer-1873}
who also gave a linear-time algorithm.
Proving that lose-in-1 checkers is an NP problem instead of, say, PSPACE or
EXPTIME involves extending this mate-in-1 result.

\paragraph{Complexity of cooperating.}
The final checkers variants we consider in this paper are \emph{cooperative}
versions, where the two players cooperate to achieve a common goal,
effectively acting as a single player.  The same proofs described above also
show NP-completeness of \emph{cooperative win-in-2} checkers --- where the
players together try to eliminate all pieces of one color in two moves --- and
\emph{cooperative always-jumping} checkers --- where the players together try
to eliminate all pieces of one color using only jumping moves.
In Section~\ref{coop}, we describe these results in more detail.

\paragraph{Checkers font.}
Finally, in Section~\ref{font}, we give a checkers \emph{puzzle font}:
a series of cooperative always-jumping checkers puzzles
whose solutions trace out all 26 letters (and 10 numerals) of the alphabet.
See Figures~\ref{fig:PuzzleFont} and~\ref{fig:SolvedFont}.
We hope that this font will encourage readers
to engage with the checkers variants studied in this paper.
The font is also available as a free web app.%
\footnote{\url{http://erikdemaine.org/fonts/checkers/}}

%
%
%

\section{Checkers Rules and Variants}

First let us review the rules of American checkers / English draughts.
Two players, Black and Red, each begin with 12 pieces placed in an arrangement on the dark squares of an $8 \times 8$ checkerboard.  Black moves first, and play takes place on the dark squares only.  Initially, pieces are allowed to move diagonally forward (i.e., toward the opponent's side) to other diagonally adjacent dark squares, but not backward.  If a player's piece, an opponent's piece, and an empty square lie in a line of diagonally connected dark squares, the player must ``jump'' over the opponent's piece and land in the empty square, thereby capturing the opponent's piece and permanently removing it from the board.  Furthermore, jumps must be concatenated if available, and a player's turn does not end until all possible connected jumps have been exhausted.  If a piece eventually reaches the last row on the opponent's side of the board, it becomes a ``king,'' which enables movement to any diagonally adjacent dark square, regardless of direction.  The goal of each player is to capture all of the opponent's pieces.

As in \cite{robson1984n,fraenkel1978complexity}, we generalize the game to an
arbitrary configuration of an $n \times n$ board.
We still assume that the next player to move is Black.

Next we describe the two main variants of checkers considered in this paper.
First, in \emph{lose-in-1} checkers, one player tries to make a single move
such that, on the following turn, the opponent is forced to capture all of the
first player's remaining pieces with a sequence of jumps.
Second, in \emph{always-jumping} checkers,
we add one rule: a piece must jump an opponent's piece on every move.
In other words, pieces are now forbidden from moving to another square without
a jump; the first player with no available jumps (or no pieces) loses.
As in standard checkers, jumps must be concatenated within a single turn as
long as there is another jump directly available.

In Section~\ref{coop}, we define and analyze cooperative versions of these two
checkers variants, where players work collaboratively to make one of them to
win (without particular regard to \emph{who} wins).  These variants arise
if the players hate draws, or if they want to make the game end quickly.

To improve readability in the following proofs and figures, we rotate the
checkerboard $45^\circ$ so that pieces move horizontally and vertically,
rather than diagonally.

\section{Lose-in-1 Checkers is NP-Complete} \label{lose1}

In this section, we prove the following computational decision problem
NP-complete:

\newcommand{\clio}{\PROBLEM{Lose-in-1 Checkers}}

\begin{problem}[\clio]
  Given a checkers game configuration, does the current player have a move such that the other player must win on their next move?
\end{problem}

In general, proving NP-completeness consists of two steps \cite{6.890}:
(1)~proving ``NP-hardness'' (Theorem~\ref{clio NP-hard} below) and
(2)~proving ``membership in NP'' (Theorem~\ref{in NP} below).
The first step establishes a \emph{lower bound} on complexity --- technically,
that the problem is as hard as all problems in the complexity class in NP.
The second step establishes an \emph{upper bound} on complexity --- that
the problem is no harder than NP.

Without getting into the definition of NP, we can prove NP-hardness of a
problem by showing that it is at least as hard as a known NP-hard problem.
We prove such a relation by giving a \emph{reduction} that shows how to
(efficiently) convert the known NP-hard problem into the target problem.
In our case, we will reduce from the following problem;
refer to Figure~\ref{fig:example-graph} for an example.


\newcommand{\feh}{\PROBLEM{Planar Forced Edge Hamiltonicity}}

\begin{problem}[\feh]
  Given a max-degree-3 planar undirected graph%
  \footnote{A graph is \emph{max-degree-3} if every vertex has at most
    three incident edges.
    A graph is \emph{planar} if it can be drawn in the plane,
    with vertices as points and edges as curves (or straight lines),
    such that the edges intersect only at common endpoints.}
  and a perfect matching in that graph,%
  \footnote{A \emph{perfect matching} in a graph is a set of edges (``matched
    edges'') such that every vertex is the endpoint of exactly one matched edge.}
  is there a cycle that visits every edge in the matching?
  
  The graph and matching must also satisfy the promise that we can draw a
  cycle that crosses every matched edge (possibly multiple times) without
  crossing any nonmatched edges.  (An equivalent promise is that the nonmatched
  cycles do not enclose any vertices of the graph,
  i.e, form faces of the graph.)
\end{problem}

\begin{figure}
  \centering
  \subcaptionbox{A graph with a matching (drawn in red).}{\includegraphics[width=.3\textwidth]{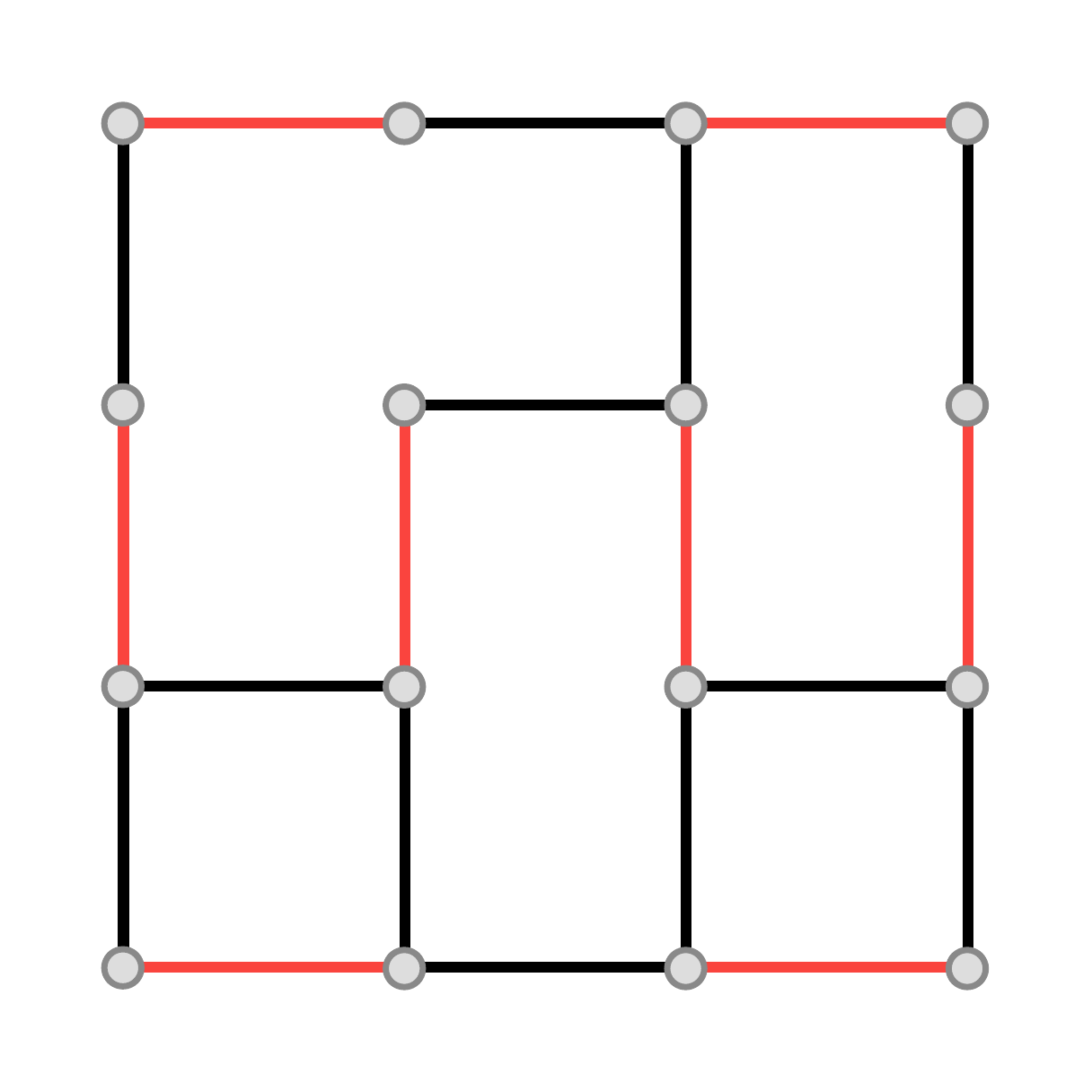}}\hfil\hfil
  \subcaptionbox{A cycle crossing all and only matched edges.}{\includegraphics[width=.3\textwidth]{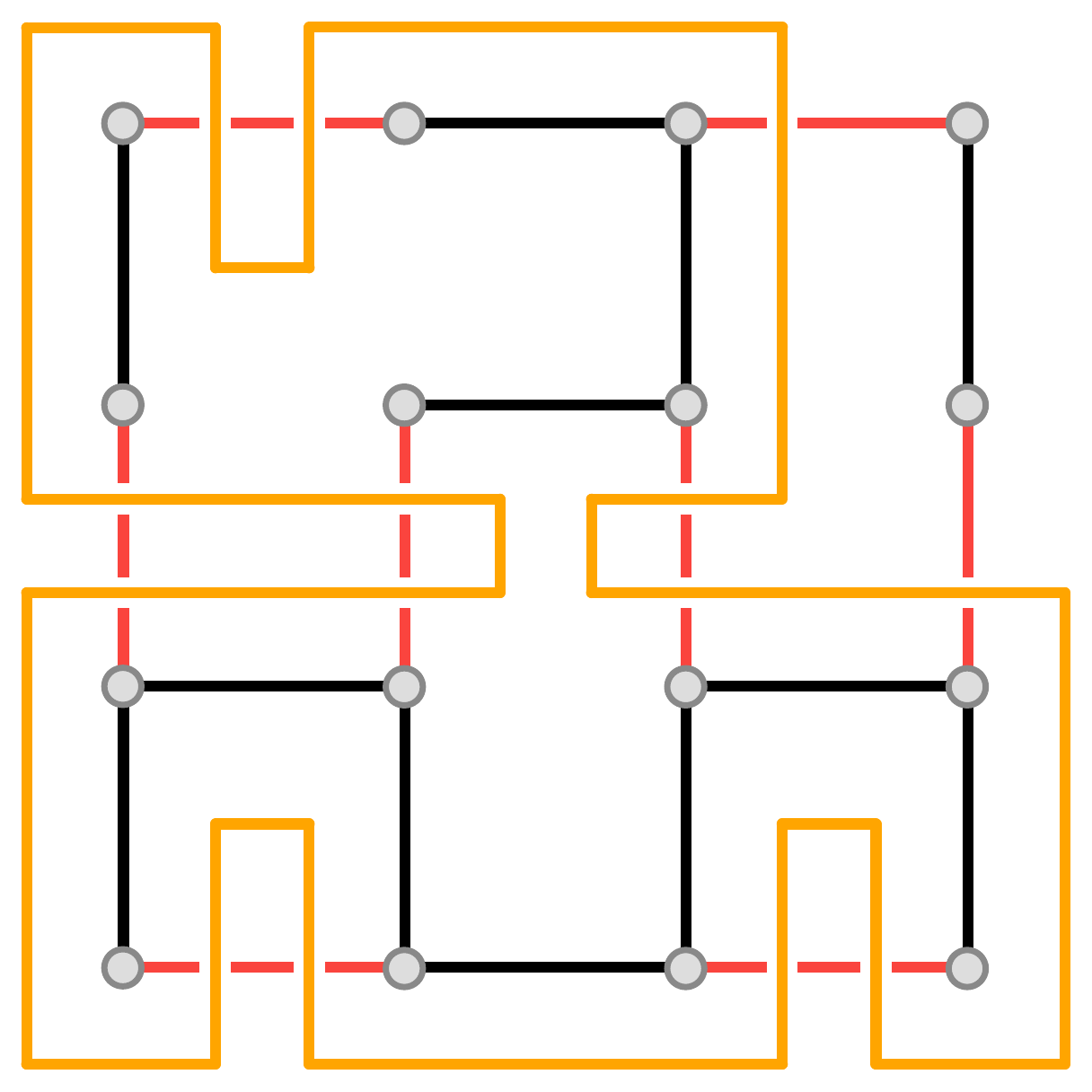}}\hfil\hfil
  \subcaptionbox{A Hamiltonian cycle in the graph covering the matched edges.}{\includegraphics[width=.3\textwidth]{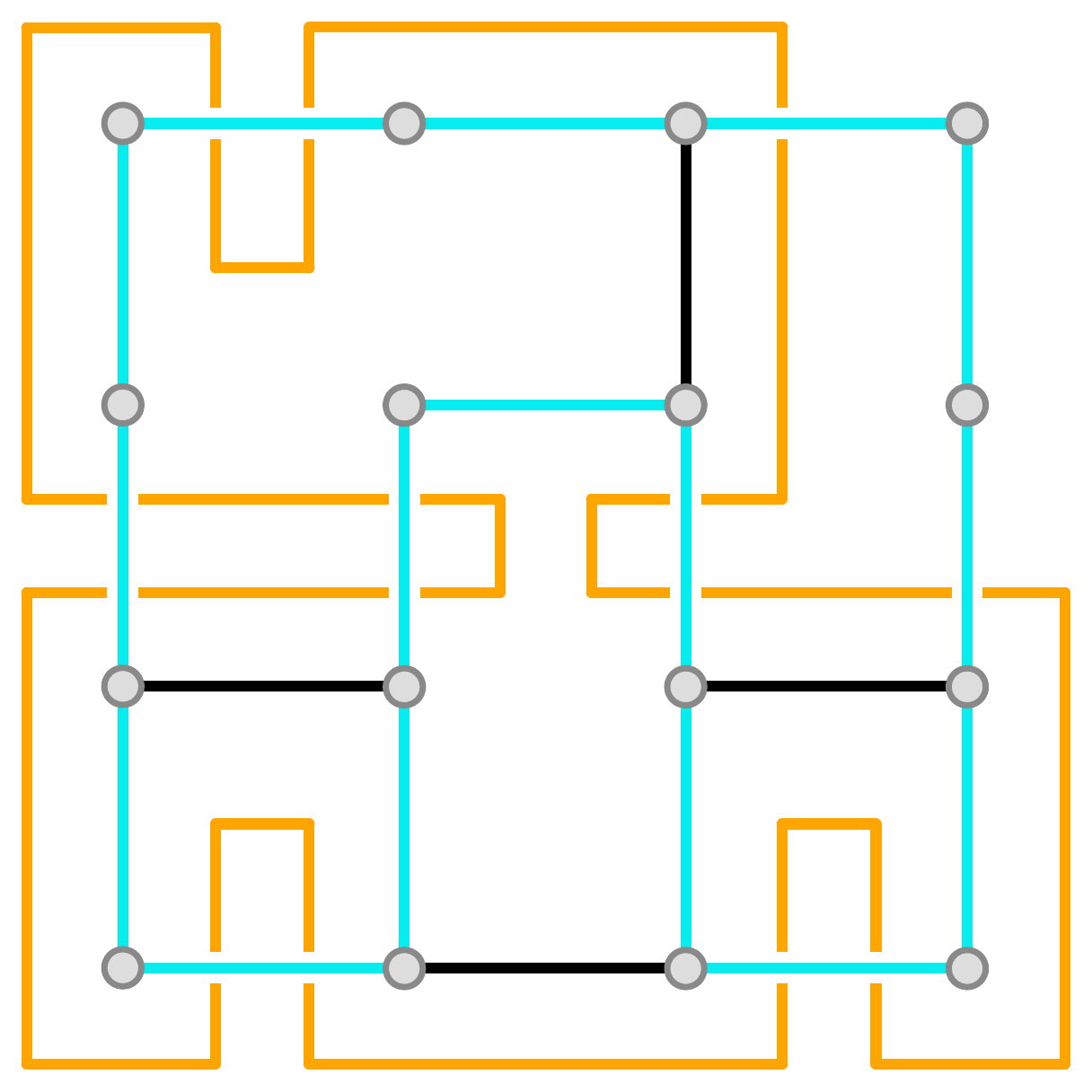}}
  \caption{An example instance of \feh\ (a) satisfying the promise (b),
    and its solution (c).}
  \label{fig:example-graph}
\end{figure}

This particular problem has not yet been shown to be NP-hard, but it follows from
an existing proof (another reduction from another known NP-hard problem):

\begin{theorem} \label{feh NP-c}
  \feh\ is \NP-complete.
\end{theorem}

\begin{proof}
This result follows from Plesn\'ik's proof \cite{PLESNIK1979} of NP-hardness
of Hamiltonicity in directed graphs where every vertex either (1)~has two
incoming edges and one outgoing edge or (2)~has one incoming edge and two
outgoing edges.
Any edge from a Type-1 vertex to a Type-2 vertex is forced to be in the
Hamiltonian cycle.
By inspecting Plesn\'ik's reduction, no two vertices of
the same type are adjacent (the two types form a bipartition).
Thus the forced edges form a perfect matching in the graph.
By further inspection, we can verify that the reduction always produces
graphs that satisfy the promise.
(The promise is not true in general; see
Figure~\ref{fig:promise_violated_matched_graph} for a counterexample.)
\end{proof}

\begin{figure}
  \centering
  \includegraphics[scale=0.4]{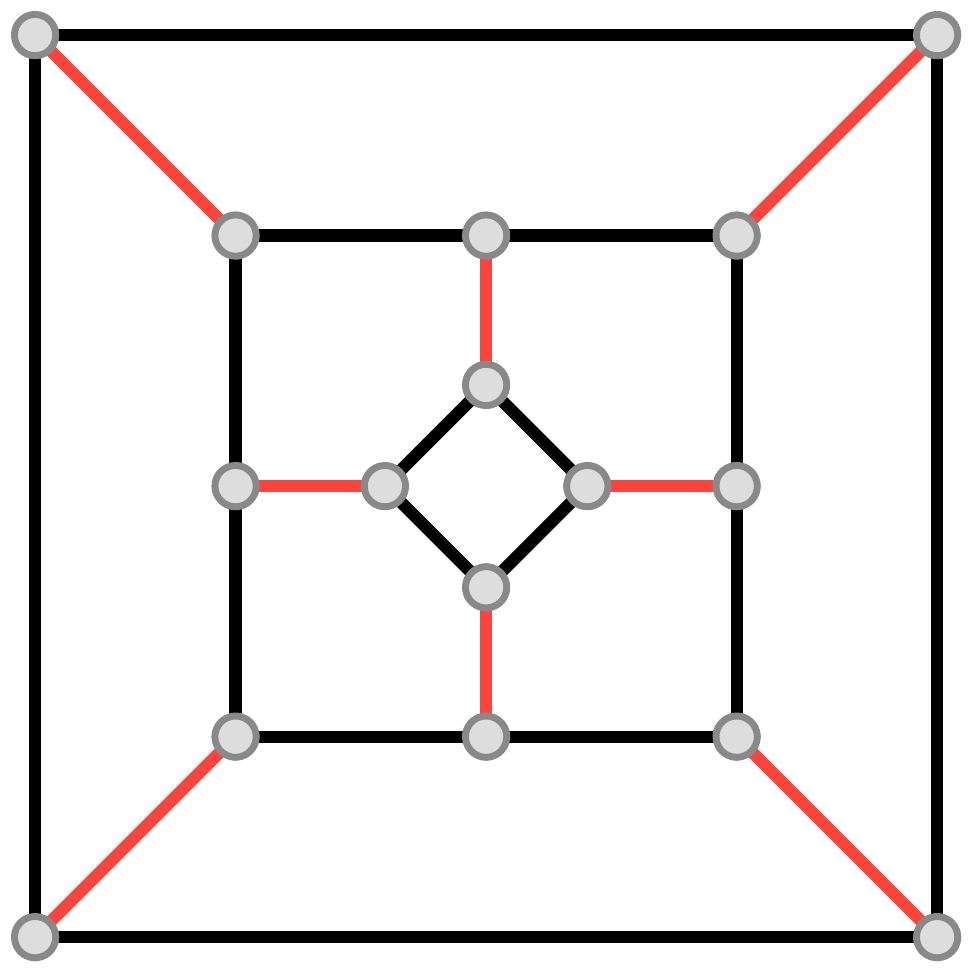}
  \caption{A graph and matching not satisfying the promise of \feh:
    the middle square of unmatched edges forms a separator that prevents any
    cycle from visiting both the diagonal and orthogonal matched edges.}
  \label{fig:promise_violated_matched_graph}
\end{figure}

Now we give a reduction from \feh\ to \clio:

\begin{theorem} \label{clio NP-hard}
  \clio\ is \NP-hard.
\end{theorem}

\newcommand{\green}{cyan}
\newcommand{\purple}{orange}

\begin{proof}
Our goal is to convert an instance of \feh\ into an equivalent instance of
\clio.  For example, the instance in Figure~\ref{fig:example-graph} will
convert into the checkers board in Figure~\ref{fig:overall reduction}.

\begin{figure}
\centering
\includegraphics[width=\textwidth]{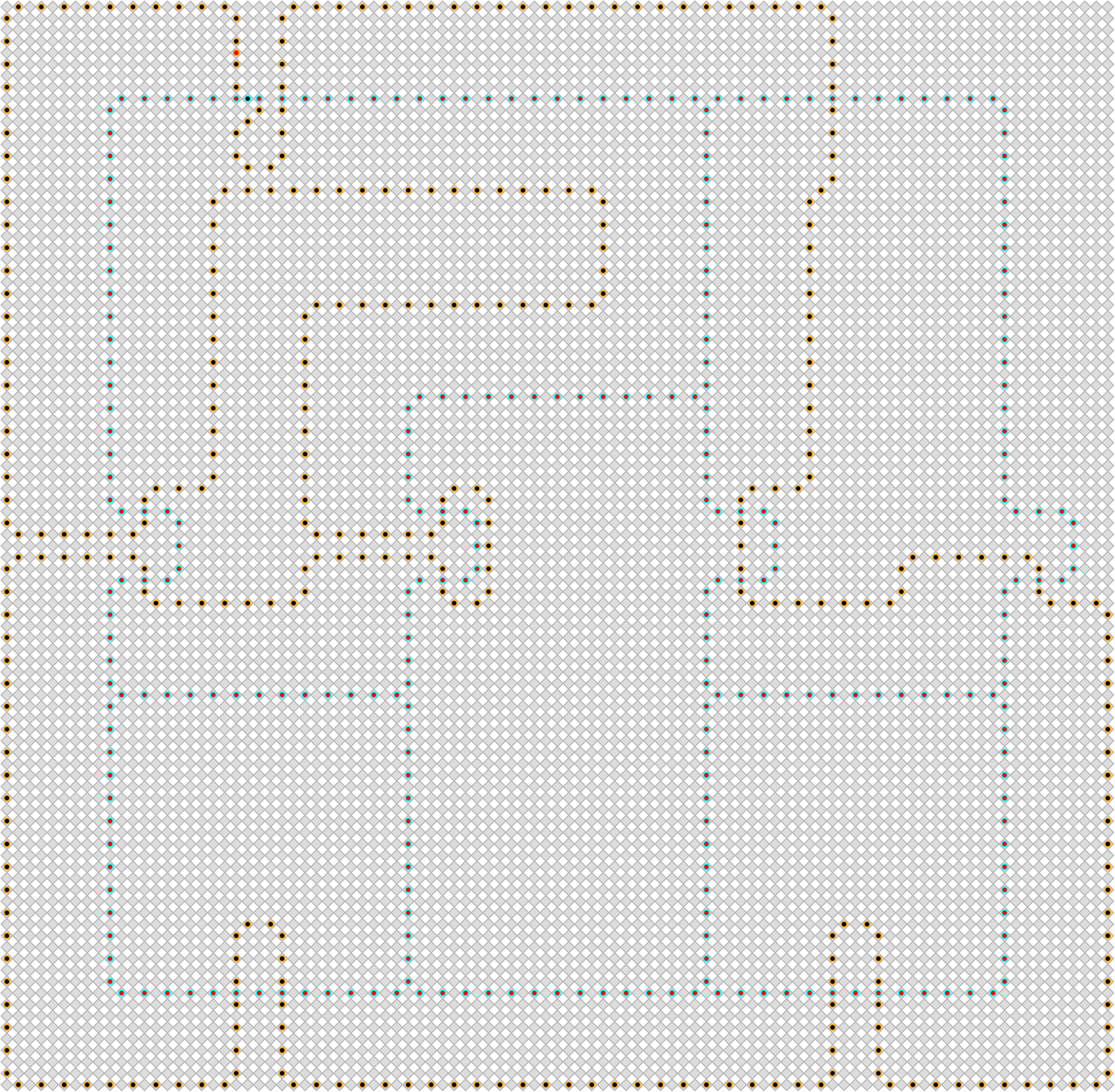}
\caption{The checkerboard produced by the reduction for the example graph and cycle in Figure~\ref{fig:example-graph}.  The motive pieces are near the upper-left horizontal edge.}
\label{fig:overall reduction}
\end{figure}

The overall structure of the reduction is to create two parity classes on the checkerboard, where each class contains one piece of one player and all other pieces of the other player.  If Black can find a Hamiltonian cycle of concatenated jumps on its turn, then Red will be forced to win by capturing all of the remaining black pieces in the second parity class; otherwise, some black pieces will remain because Red cannot traverse the entire cycle.

To produce a checkers instance from a \feh\ instance, we begin by embedding the 3-regular planar undirected graph in the planar grid.  Theorem 1 of \cite{StorerEmbedding} establishes that a graph with $n$ vertices can always be embedded in an $n\times n$ grid.  We scale the grid embedding by 27 and place red pieces at every second point along the edges.  We modify the vertical matched edge segments by moving the middle three pieces in each segment six columns to the right, then adding three red pieces in the same parity class in the rows immediately above the upper moved piece and below the lower moved piece.  We call the resulting structures \emph{tabs} (by analogy with the tabs and pockets of jigsaw puzzles).  In the figures, this grid embedding consists of the red pieces on \green{} squares.

The forced edge Hamiltonicity instance includes the promise that we can draw a cycle that crosses all matched edges and no nonmatched edges.  We construct such a cycle using the following algorithm:
\begin{enumerate}
\item Take the subgraph of the dual graph%
  \footnote{The dual graph has a (dual) vertex for each face, and a (dual) edge
    between two (dual) vertices corresponding to faces that share an edge.}
  given by the duals of the matched edges.
  By the promise, this graph is connected.
\item Until this graph is a tree, choose an edge in a cycle and
  replace one endpoint of the chosen edge with a new vertex.
\item Double the tree to form a cycle.
\end{enumerate}

Each matched edge has \emph{cycle-crossing points} at which the cycle may cross.  For the horizontal edges, the cycle-crossing points are the outer two pieces of the three center pieces in the edge; for the tabs, they are the middle piece of both horizontal sections of each tab.  We lay out this cycle in the checkerboard starting from an edge that is crossed twice consecutively (there is at least one such edge, corresponding to a leaf in the doubled tree).  At each such immediate recrossing, after entering at a cycle-crossing point, the cycle steps into the face for a limited distance, then leaves via the other cycle-crossing point (see Figure~\ref{fig:recrossing}).  At other crossings, the cycle walks the boundary of the face starting with the adjacent edge nearest the cycle-crossing point, inset by 7 spaces from horizontal edges and 8 from vertical edges (being off-by-one is necessary for the parity to work out).  This inset ensures the cycle can enter and leave at immediate recrossings without interfering with the boundary walk.  The walk crosses subsequent edges using the first cycle-crossing point encountered (and possibly resumes in the same direction after reentering via the other cycle-crossing point on that edge).  The cycle walks the outer face in the same way as any other face.  We place black pieces at every second point on the cycle, in the opposite parity class from the red empty pieces (so they are adjacent, not coincident, at edges).  In the figures, these are the black pieces on \purple{} squares.

\def\scale{1.5}

\begin{figure}
  \centering
  \subcaptionbox{Horizontal}{\includegraphics[scale=\scale]{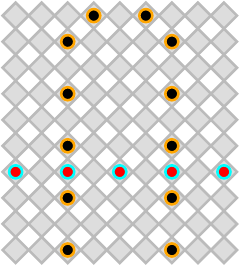}}\hfil\hfil
  \subcaptionbox{Tab, from left}{\includegraphics[scale=\scale]{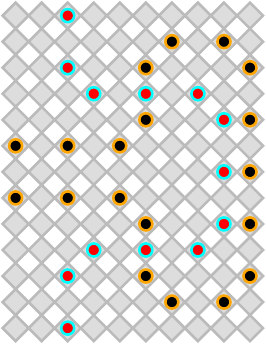}}\hfil\hfil
  \subcaptionbox{Tab, from right}{\includegraphics[scale=\scale]{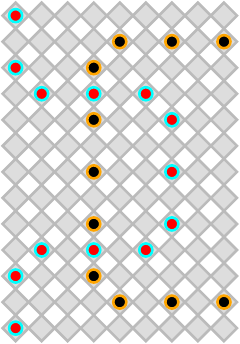}}
  \caption{Behavior of the cycle when immediately recrossing edges.}
  \label{fig:recrossing}
\end{figure}

We now place the \emph{motive} pieces, one per player, which will be the pieces moving during their turns.
Only the motive pieces are positioned to make captures, so because captures must be made if available, only the motive pieces will move during the next two plies.  We place Black's motive piece between two red pieces in the \green{} parity class along a forced edge, where one of the adjacent red pieces is the point of a cycle crossing. The cycle crossing is moved one unit to now cross at the black motive piece; we can always do this without interfering with the rest of the instance because the crossings were placed at least three spaces apart in the previous step.  (See Figure~\ref{fig:motive}.)  Red's motive piece is then placed between two black pieces in the \purple{} parity class.  Both motive pieces are kings, allowing them to move in any direction while traversing the graph or cycle.

\begin{figure}
  \centering
  \subcaptionbox{Horizontal}[1in]{\includegraphics[scale=\scale]{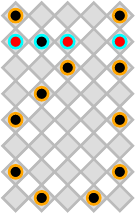}}\hfil\hfil
  \subcaptionbox{Tab}{\includegraphics[scale=\scale]{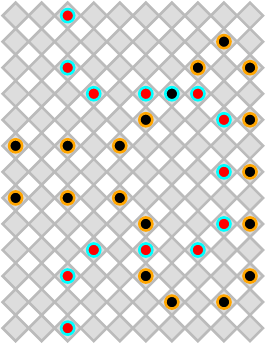}}
  \caption{Placing the black motive piece.}
  \label{fig:motive}
\end{figure}

If the input forced edge Hamiltonicity instance contains a Hamiltonian cycle including all the forced edges, then Black can capture with the motive piece all red pieces along all forced edges in the \green{} parity class, with the motive piece returning to its initial position.  Red is then forced to capture the entire \purple{} cycle, which contains all black pieces (including the black motive piece), and win the game.  Thus the lose-in-1 instance has an answer of \textsc{yes}.

If the input forced edge Hamiltonicity instance does \emph{not} contain a Hamiltonian cycle including all the forced edges, Black will not be able to capture any red pieces on at least one of the forced edges in the \green{} parity class.  Then Red's capture sequence will end because the red piece at a cycle crossing of that forced edge will block its motive piece from continuing to capture, so Red will not win the game with this move.  Thus the lose-in-1 instance has an answer of \textsc{no}.
\end{proof}

The second step is to show membership in NP.  The goal here is to show that,
when the answer to the problem is \textsc{yes}, then there is a succinct
certificate (essentially, a solution) that easily proves that the answer
is \textsc{yes}.

\begin{theorem} \label{in NP}
\clio\ is in \NP.
\end{theorem}

\begin{proof}
The certificate is Black's move (capture sequence) --- whose length
is less than the number of pieces --- that forces Red to win.
We give an efficient (polynomial-time) algorithm to verify such a certificate.
First we execute the capture sequence to recover the board state at the beginning of Red's turn.
Now we need to decide whether Red has the option to \emph{not} win in one move,
in which case the certificate is invalid.
If Red cannot capture, then this is the case: Red either cannot move at all
(and thus loses) or any move does not immediately win.
Thus assume that Red can capture (otherwise return \textsc{invalid}).
For each red piece $r$ that can capture, we will check whether moving $r$
can lead to Red ``getting stuck'', leaving some black pieces uncaptured.

Construct the graph for solving mate-in-1 from \cite{fraenkel1978complexity}
as described in Section~\ref{mate-in-1}, with vertices representing
reachable positions for $r$ by jumping and edges representing (eventually)
jumpable black pieces.  For the certificate to be valid, there must be
at least one winning move using $r$, so the graph must include all black pieces
as edges, and it must have an Eulerian path starting at~$r$
(otherwise return \textsc{invalid}).
Thus the graph has zero or two vertices of odd degree, and if there are two,
one of them must be~$r$.

If all vertices of the graph have even degree, then check whether the graph
minus the starting vertex for $r$ has any cycles.  If it does have cycles,
then Red can avoid that cycle and cover the rest of the graph, and thereby get
stuck at the start vertex.  Conversely, if Red can get stuck, then there must
be a cycle in the remainder, because the graph is Eulerian.
Therefore we can detect validity in this case.

If the graph has two vertices $r,s$ of odd degree, then connect them both to a
new vertex~$r'$, and apply the algorithm for the previous case
starting from~$r'$.  Cycles in this graph including $r'$ are equivalent to
paths in the old graph starting at $r$ and ending at~$s$.  Therefore we can
again detect validity.
\end{proof}

\section{Always-Jumping Checkers is PSPACE-Complete} \label{always-jumping}

In this section, we prove the following computational decision problem
PSPACE-complete:

\newcommand{\ajc}{\PROBLEM{Always-Jumping Checkers}}

\begin{problem}[\ajc]
Given a checkers game configuration in which all moves must capture a piece, does the first player have a winning strategy?
\end{problem}

\begin{theorem}
\ajc\ is \PSPACE-complete.
\end{theorem}

Again, the proof consists of two parts --- PSPACE-hardness and membership in
PSPACE.  In this case, however, membership in PSPACE is straightforward,
because \ajc\ falls into the class of bounded 2-player games \cite{GPC}: each
move strictly decreases the number of pieces, so the number of moves is
bounded by the initial number of pieces.  To prove PSPACE-hardness, we reduce from
a known PSPACE-hard problem called \PROBLEM{Bounded 2-Player Constraint Logic}
(B2CL) \cite{GPC}.  For this proof, we assume familiarity with~\cite{GPC}
or the corresponding lectures of~\cite{6.890}.

\begin{proof}
To reduce from B2CL, we construct \textsc{variable}, \textsc{fanout}, \textsc{choice}, \textsc{and}, and \textsc{or} gadgets in always-jumping checkers that simulate the corresponding gadgets from B2CL.
Refer to Figure~\ref{fig:B2CLGadgets}.

\paragraph{Wire.} A wire is simply a line of alternating red pieces and empty spaces.  A black piece can traverse the entire length of a wire during a single move.  As every piece is a king, a wire can turn at any empty square along its length.

\paragraph{Parity shift.} During a single turn, a piece is confined to one of four parity classes (one of the four squares in a $2 \times 2$ block of squares).  In most of our gadgets, the parity of the incoming and outgoing pieces differs.  We handle this with a \textsc{shift} gadget, as shown in Figure~\ref{fig:shift} (which can also be seen as removing one of the output wires from a \textsc{choice} gadget).  Repeated applications of the \textsc{shift} gadget (and its $90^\circ$ rotation) can move a piece to any parity class of the board, so a combination of wires, turns, and shifts can move a piece to an arbitrary dark square.

\begin{figure}[h]
\centering
  \def\scale{1.05}
  \subcaptionbox{\label{fig:shift}\textsc{shift}}{\includegraphics[scale=\scale]{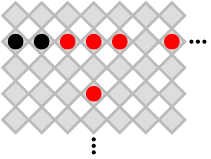}}\hfill
  \subcaptionbox{\label{fig:variable}\textsc{variable}}[0.9in]{\includegraphics[scale=\scale]{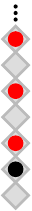}}\hfill
  \subcaptionbox{\label{fig:fanout}\textsc{fanout}}{\includegraphics[scale=\scale]{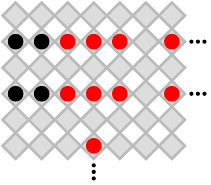}}\hfill
  \subcaptionbox{\label{fig:choice}\textsc{choice}}{\includegraphics[scale=\scale]{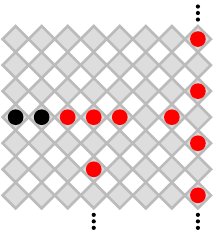}}\hfill
  \subcaptionbox{\label{fig:and}\textsc{and}}{\includegraphics[scale=\scale]{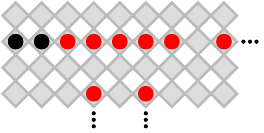}}\hfill
  \subcaptionbox{\label{fig:or}\textsc{or}}{\includegraphics[scale=\scale]{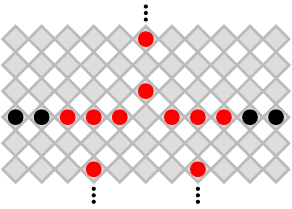}}
  \caption{Gadgets in our reduction from Bounded 2-Player Constraint Logic
    \cite{GPC} to Always-Jumping Checkers.  For example, black can get a piece to
    the top output of the \textsc{or} gadget only if they can get a piece to
    one of the inputs,
    and black can get a piece to the top output of the \textsc{and} gadget
    only if they can get a piece to both of the inputs.
    Gadgets (b)--(e) also form a reduction from Bounded Nondeterministic
    Constraint Logic \cite{GPC} to cooperative always-jumping checkers.}
  \label{fig:B2CLGadgets}
\end{figure}

\paragraph{\textsc{variable}.} When the game begins, players take turns choosing to either activate or deactivate the \textsc{variable} gadgets, shown in Figure~\ref{fig:variable}.  If Black moves first on a given variable, he must continue capturing red pieces along the wire, corresponding to setting the variable to true.  If Red moves first, he can capture the black piece, corresponding to setting the variable to false.  Once set, a variable never changes value.

\paragraph{\textsc{fanout}, \textsc{choice}, \textsc{and}, \textsc{or}.} These gadgets are all based on a similar construction.  Triples of red pieces prevent the black pieces from capturing, and pairs of black pieces prevent the red pieces from capturing.  When a black piece arrives via an input wire, it captures the middle piece of some triple(s), allowing other black pieces to leave the gadget along the output wires.  In the \textsc{fanout} gadget, the arriving black piece breaks both triples, allowing a black piece to exit along each output wire.  \textsc{choice} also has two output wires, but only one black piece can leave, so Black must choose which wire to activate.  In the \textsc{and} gadget, both black pieces must arrive to break their respective triples, while in the \textsc{or} gadget, the arrival of a black piece along either wire enables a black piece to leave.  Even if both \textsc{or} inputs are activated, the output wire is consumed when the first black piece leaves.

\paragraph{Board layout.} Our reduction translates the given B2CL graph directly using the gadgets above, adding turns and shifts as necessary for layout.
Black is the first player to move in the resulting checkers game.

In B2CL, the first player wins if they can flip a particular edge (specified
as part of the instance), which corresponds to activating an \textsc{and}
output wire in the checkers game.  That \textsc{and} output wire leads to a
long series of \textsc{shift} gadgets, giving Black a large number $F_B$ of
free moves if Black can activate the \textsc{and} output.

In B2CL, the Red player needs to be able to pass, so we also provide Red with
a large number $F_R$ of free moves using a collection of $F_R$ isolated
gadgets consisting of a pair of red pieces adjacent to a single black piece:
\includegraphics{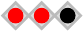}.
(Recall that the player to move loses if they cannot capture.)
We set $F_R$ to twice the number of gadgets not counting the long chain of
\textsc{shift} gadgets, which is an upper bound on the number of black moves
if the formula does not get satisfied.
We then set $F_B = 2 \, F_R$, so that Red runs out of
moves if Black can activate the desired series of \textsc{shift} gadgets.
%
Therefore Black wins the checkers game exactly when the first player wins the B2CL instance, so \ajc\ is \PSPACE-complete.
\end{proof}

\section{Cooperative Versions are NP-Complete} \label{coop}

In cooperative lose-in-1 (win-in-2) checkers, players collaborate to force one of them to win in two moves.  The computational decision problem for this game asks whether the current player has a move such that the other player can win on their next move, which is almost identical to the competitive lose-in-1 decision problem.  Indeed, the same reduction as in Section \ref{lose1} shows cooperative lose-in-1 (win-in-2) checkers is \NP-complete.

In cooperative always-jumping checkers, players collaborate to remove all the pieces of one color, and so the game becomes essentially a single-player puzzle.  Because at least one piece is still consumed on every turn, the game is bounded, so it is in NP. Furthermore, we can prove NP-hardness by reducing from \PROBLEM{Bounded Nondeterministic Constraint Logic} (Bounded NCL), which is \NP-complete \cite{GPC}.  It happens that the same \textsc{fanout}, \textsc{choice}, \textsc{and}, and \textsc{or} gadgets from Section~\ref{always-jumping} can form this reduction, which proves that cooperative always-jumping checkers is \NP-complete.

\section{Checkers Font} \label{font}

In Figure~\ref{fig:PuzzleFont}, we show a collection of cooperative
always-jumping checkers puzzles whose solutions trace letters of the alphabet
and digits 0--9, as shown in Figure~\ref{fig:SolvedFont}.
Each puzzle is a plausible end game of cooperative
always-jumping checkers.  In each puzzle, Black moves first and, in any
solution that eliminates all pieces of one color, the paths traced by
the movement of the pieces combine to create the letter or digit.
In some puzzles, Black wins (e.g., `b' and `w'), while in other games,
Red wins (e.g., `r' and~`y').

\begin{figure}
\centering
  \includegraphics[width=\linewidth]{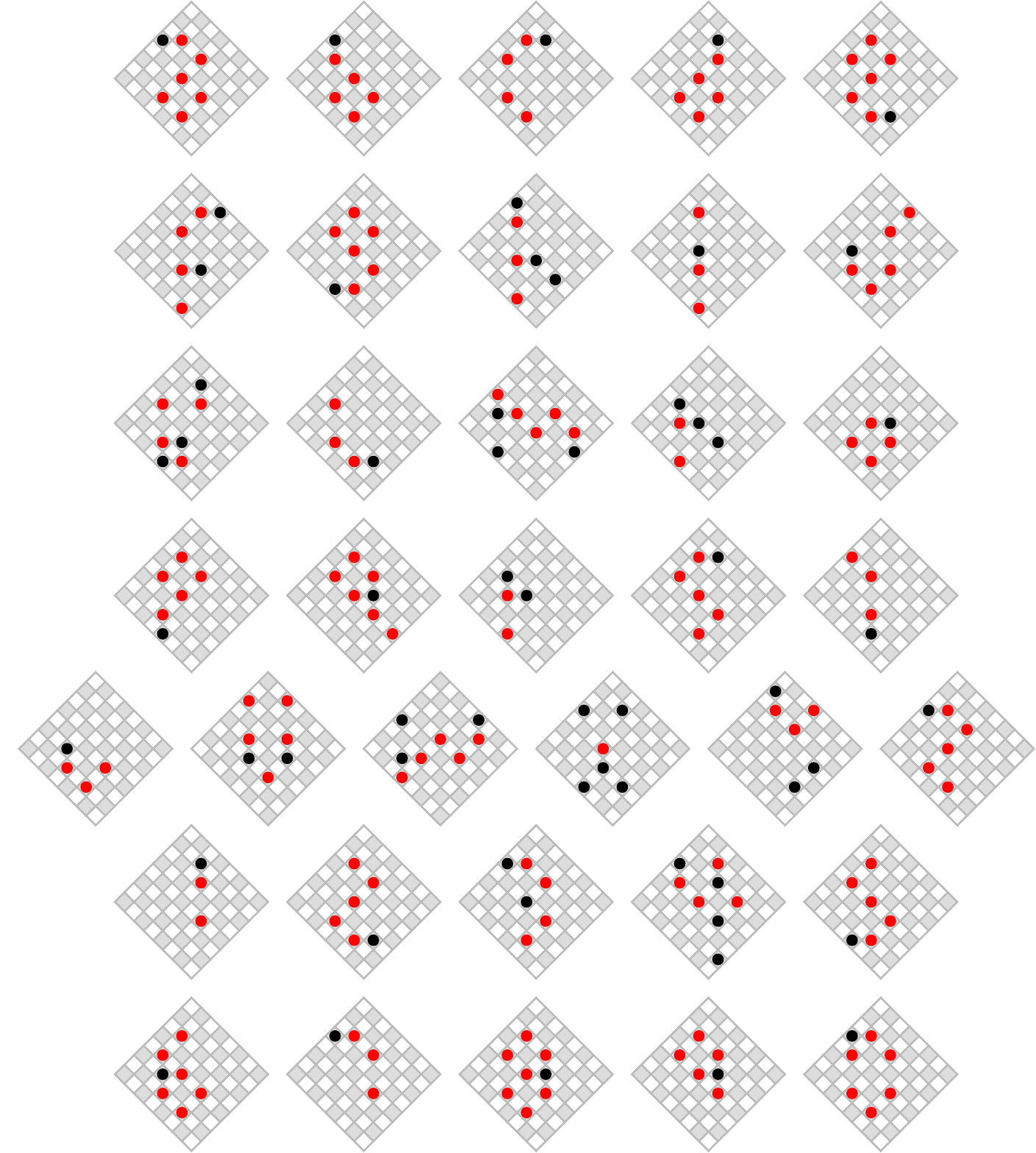}
  \caption{A puzzle font for checkers.  Each letter and digit is represented by
    an $8 \times 8$ cooperative always-jumping checkers puzzle,
    where the goal is to eliminate all pieces of one color.
    Black always moves first, but does not always win.
    The movement of the pieces in a puzzle solution
    traces the letter or digit, as revealed in Figure~\ref{fig:SolvedFont}.
    See \url{http://erikdemaine.org/fonts/checkers/} for an interactive version.}
  \label{fig:PuzzleFont}
\end{figure}

\begin{figure}
\centering
  \includegraphics[width=\linewidth]{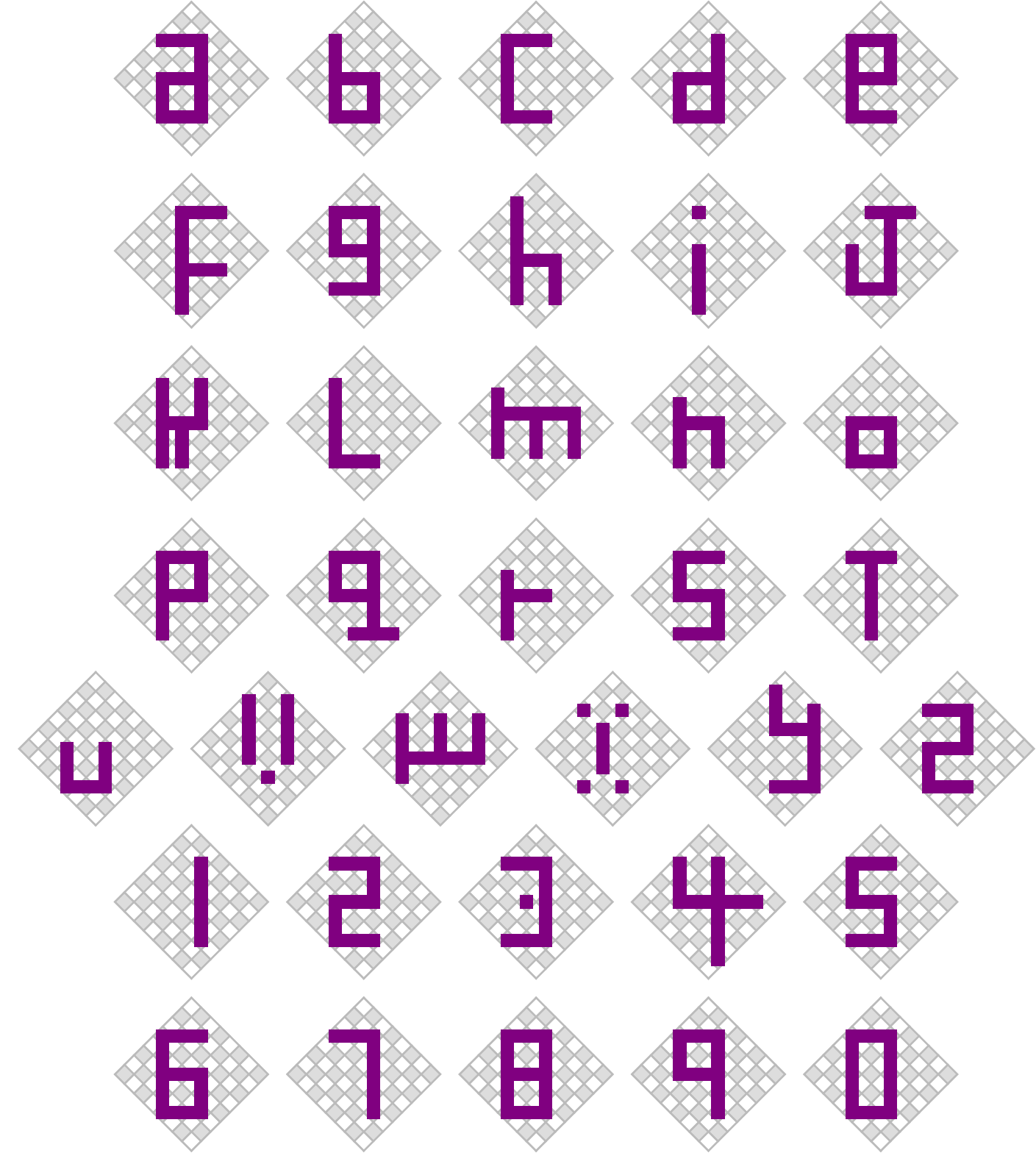}
  \caption{Solved font: Paths traced by the solutions to the $8 \times 8$ cooperative always-jumping checkers puzzles in Figure~\ref{fig:PuzzleFont}.}
  \label{fig:SolvedFont}
\end{figure}

\section{Open Problems}

The main remaining open problem is to analyze the computational complexity
of suicidal (mis\`ere) checkers.  A natural conjecture is that this game is
EXPTIME-complete, like normal checkers.

Another natural open problem is to analyze the computational complexity of
cooperative checkers, where the players together try to eliminate all pieces
of one color.  We have shown that this problem is NP-complete when limited
to just two moves or when we add the always-jumping rule.

A final open problem is the computational complexity of \emph{mate-in-2}
checkers: making a move such that, no matter what move the opponent makes,
you can win in a second move.  Given that mate-in-1 checkers is easy,
while lose-in-1 is NP-complete, we might naturally conjecture that
mate-in-2 checkers is NP-complete, but our proofs do not immediately apply.

\section*{Acknowledgments}

This work was done in part during an open problem session about hardness
proofs, which originally grew out of an MIT class (6.890, Fall 2014).
We thank the other participants of the session, in particular Mikhail Rudoy
who proved Theorem~\ref{in NP} and allowed us to include his proof.

\let\realbibitem=\bibitem
\def\bibitem{\par \vspace{-1.2ex}\realbibitem}

\bibliography{references}
\bibliographystyle{alpha}

\end{document}